\documentclass[11pt, twoside, letter]{article}
\usepackage[margin = 1in]{geometry}
\usepackage[utf8]{inputenc}
\usepackage{amsmath, amssymb, amsfonts}
\usepackage{ifthen}
\usepackage[boxed, linesnumbered]{algorithm2e}
\usepackage{hyperref}       

\usepackage{pgf}
\usepackage{tikz}
\usetikzlibrary{arrows,automata}

\usepackage{microtype}

\usepackage[utf8]{inputenc}
\usepackage{amsmath, amssymb, amsfonts}
\usepackage{ifthen}
\usepackage{comment}

\usepackage{pgf}
\usepackage{tikz}
\usetikzlibrary{arrows,automata}

\usepackage{amsthm}

\usepackage{mathtools}

\newtheorem{theorem}{Theorem}[section]
\newtheorem{corollary}[theorem]{Corollary}

\newtheorem{definition}[theorem]{Definition}
\newtheorem{remark}[theorem]{Remark}
\newtheorem{proposition}[theorem]{Proposition}

\newtheorem{assumption}[theorem]{Assumption}

\newcommand{\Real}{\mbox{${\mathbb R}$}}

\newcommand{\norm}[1]{\left\lVert #1 \right\rVert}

\let\norm\undefined 
\DeclarePairedDelimiter\norm{\lVert}{\rVert}

\title{Learning Agents in Black-Scholes Financial Markets}

\author{Tushar Vaidya\\SUTD\\tushar\_vaidya@sutd.edu.sg
	\and Carlos Murguia\\ TU/e \\C.G.Murguia@tue.nl
	\and Georgios Piliouras\\SUTD\\georgios@sutd.edu.sg}

\begin{document}


\maketitle
\begin{abstract}
Black-Scholes (BS) is the standard mathematical model for European option pricing in financial markets. Option prices are calculated using an analytical formula whose main inputs are strike (at which price to exercise) and  volatility. The BS framework assumes that volatility remains constant across all strikes, however, in practice it varies. How do traders come to learn these parameters?

We introduce natural agent-based models, in which they update their beliefs about the true implied volatility based on the opinions of other traders. We prove exponentially fast convergence of these opinion dynamics using techniques from control theory and leader-follower models, thus providing a resolution between theory and market practices. We allow for two different models, one with feedback and one with an unknown leader.
\end{abstract}

\section{Introduction}
Econophysics divides into two paradigms. Statistical econophysics relies on data, fitting certain power laws to existing asset prices at various time scales \cite{schinckus2012methodological, chakraborti2011econophysicspart1}. In statistical econophysics, zero-intelligence agents have random interactions. Agents are homogenous and have no learning ability. The central object of study is historical price data. The viewpoint is that interacting zero-intelligence traders' actions are already incorporated into price fluctuations. The focus is on the macroscopic aggregation in the form of available data. While this is an important area of research, agent-based models offer the opportunity to study the microscopic interactions in more detail.  Here agents are heterogeneous.

Our objective is to offer a cogent and clear motivation for agent-based econophysics in the context of option volatilities, whereby learning and interaction are made explicit. To an outsider it may seem that financial assets are observed at one price, decided by the market. In reality, prices fluctuate throughout the day and there is no equilibrium price: it is always in flux. Interaction between strategic traders and other players is embedded in all transactions and informational channels. Interaction is vital to understanding markets.  The motivation for this paper was inspired by the works of Kirman and Follmer \cite{kirman2002reflections, follmer2005equilibria}. Rather than develop a full-blown game theoretic or mean-field model, we advocate something in between where interaction of traders is intrinsic. We aim to take a more nuanced view of agent-based econophysics as espoused by  Chakroborti et al. \cite{chakraborti2011econophysicspart2}.

Most trading is done electronically. To be dominant, firms now invest huge sums in technology to get an edge. For futures trading, speed is vital to profits. Trading complex derivatives requires not only speed but huge amounts of investment in quantitative models. This in turn feeds the need for mathematicians, computer scientists and engineers. Increasingly, over the last two decades the way trading is conducted has also seen drastic changes. Electronification of the markets has affected both instruments traded on and off exchange. Algorithmic trading drives not only plain vanilla instruments like stocks and futures but plays a crucial role in derivatives trading \cite{bacoyannis2018idiosyncrasies, ganeshmulti, wei2019model}. Furthermore, the distinction between stock exchanges and over-the-counter (OTC) markets is not as clear as it once was \cite{malamud2017decentralized, das2008effects, simaan2007price}. In OTC markets, trading is between two counterparties and there is no centralized marketplace. Increasingly, over the last decade there has been a regulatory push to make OTC markets more exchange-like. In over-the-counter markets, participants may see what their competitors are quoting for a particular security but volume and the actual price transacted remain the privy of the bilateral counterparties. In some quarters, OTC markets are usually referred to as being quote-driven or truly dark markets \cite{duffie2011dark}. Regulation in the United States and European Union has resulted in fragmented exchange based trading but centralization of opaque OTC markets

\subsection{Options Markets}
Derivative contracts  are actively traded across the world's financial markets with a total estimate worth in the trillions of dollars.
To get an intuitive understanding of the setting and the issues at hand, let's consider the prototypical example of European options.

A European option is the right to buy or sell an underlying asset at some  point in the future at a fixed price, also known as the strike.
A call option gives the right to buy an asset and a put option gives the right to sell an asset at the agreed price.
On the opposite  side of the buyer is the seller who has relinquished his control of exercise.  Buyers of puts and calls can exercise the right to buy or sell.  Sellers of options have to fulfil obligations when exercised against.
 The payoff of a buyer of a call option with stock price $S_T$ at expiry time $T$ and exercise price $K$ is
$\max\{S_T-K,0\}$, whereas
for a put option  is
$\max\{K-S_T,0\}$.

To get a price we input the current stock price $S_0$ (e.g. \$101), the exercise price $K$ (e.g. \$90), the expiry $T$ (e.g. three months from today) and the volatility $\sigma$ in the Black-Scholes (BS) formula and out comes the answer, the quoted price of the instrument~\cite{chriss1996black, otto2001finite, kakushadze2017volatility}.
\[
 \mbox{Price}=BS(S_0,K,T,\sigma).
\]

 Volatility, which captures the beliefs about how turbulent the stock price will be, is left up to the market. This parameter is so important that in practice the market trades European calls and puts by quoting volatilities.\footnote{Using the Black-Scholes formula with particular implied volatility, traders obtain a dollar value price.}

 Options can be struck at different strike prices on the same asset (e.g. $K= \$90, \$75, \$60$).  If the underlying asset and the time to exercise $T$ (e.g. 3 months)  are the same,  one would expect the volatility to be the same at different strikes. In practice, however, the market after the 1987 crash has evolved to exhibit different volatilities.  This rather strange phenomenon is referred to as the smile, or smirk (see figure \ref{fig:smile}).
 Depending on the market, these smirks can be more or less pronounced. For instance, equity markets display a strong skew or smirk. A symmetric smile is more common in foreign exchange options markets. An excellent introduction to volatility smiles is given in \cite{derman2016volatility}.

How does the market decide about what the quoted volatility should be (e.g. for a stock index in 3 months from now) is a critical, but not well understood, question. This is exactly what we aim to study by introducing models of learning agents who update their beliefs about the volatility. Agent-based models on volatility-smile interaction and formation have not been thoroughly addressed in  finance or econophysics. They remain a challenge \cite{sornette2014physics}. Previous attempts have been made but the focus has never been on the mathematical or specific nature of interaction \cite{vagnani2009black, liu2014collective}.  Furthermore, our work takes into account the physicality of how trading occurs. An alternative perspective is offered in \cite{li2013investors, platen1998feedback}, again though the nature of interaction is missing. Nevertheless, these early attempts offer a good indication that at least the problem has garnered significant interest in different disciplines.

\begin{figure} \label{smiles}
	\centering
	\includegraphics[scale=1.2]{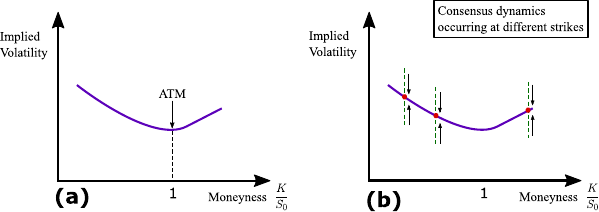}\caption{\textbf{(a)} A typical implied volatility smile for varying strikes $K$ divided by fixed spot price. Moneyness is $K/S_0$. ATM denotes at-the-money where $K$ equals $S_0$, \textbf{(b)} Consensus occurs as all investors' opinions of the implied volatility converge, round by round, to a distinct value for varying strikes.}
	\label{fig:smile}
\end{figure}

\subsection{Econophysics}
The challenge for  physicists is not to force existing physics-based models on human behaviour but rather develop new models \cite{challet2016regrets, iori2012agent, sinha2010econophysics}. To go from local microscopic interactions to global macroscopic behaviour is not an easy task \cite{stanley1996anomalous, schinckus2018ising}. In fact, the choice of models seems infinite. There are a plethora of agent-based models \cite{chakraborti2011econophysicspart2, sinha2010econophysics, castellano2009statistical}. Which one is correct?  And moreover which type of social learning is representative of financial markets trading. Barron provides an early guide \cite{lebaron2001builder}. Agent-based models were proclaimed as the future for econophysics \cite{farmer2009economy, samanidou2007agent}.  While development in this area has been steady, the problem of the emergence of volatility smiles remains unresolved. The volatility smile is an active and vigorous area of research in the mathematical finance community. Many models postulate a stochastic process for the underlying stock and volatility combined. 

\subsubsection{Knightian Uncertainty}
Risk and uncertainty are two different concepts \cite{ellsberg1961risk, knight2012risk,schinckus2009economic}. Risky assets are those on which the probabilities of random events are well-defined and known. For instance, suppose we observe historical data of a stock price. Are we confident to say we know the distribution of the stock's returns? If we are, then the stock is considered risky. Its risk is quantifiable. However, if we were unsure of even the correct probability measure, then we would be faced with uncertainty. In a sense, this captures the essence of financial markets. Traders and players use different probability measures. No such probability measure dominates. In incomplete markets, the choice of a correct probability measure such that a derivative contract is priced correctly is a subjective and quantitative exercise. In any case, no correct model exists \cite{duembgen2014estimate, khrennikova2019asset, mykland2003financial, cheridito2017duality, acciaio2016model}. As a result, participants in financial markets are free to choose whichever probability model they calibrate to market data \cite{davis2016model, cont2006model, burzoni2016universal}. 

The problem with economics based models and those in mathematical finance literature is that many times the analysis is centred on a representative agent. In case of risk and uncertainty, the choice of pricing a derivative contract boils down to choosing a correct equivalent martingale measure under which a derivative claim is replicable. For market-makers and dealers, the choice of models is vast. Each player has to make a choice and inevitably no two institutions will use the same models with the same parameters. In this case, it is remarkable that the market will aggregate the diverse beliefs to arrive at a consensus smile. At the microscopic level, though the dealers are observing each others' updates. Hence, our model can be seen as a meta opinion dynamics framework built upon the individual choices of the dealers.

\subsubsection{Financial markets: non-Bayesian}
In financial markets, updating occurs at high frequency across geographic locations \cite{wissner2010relativistic, buchanan2015physics}. Agents move simultaneously:  cancellations are the norm \cite{gu2013position, yoshimura2020mathematical, eisler2012price}. In practical terms, sequential Bayesian learning models don't seem appropriate \cite{hkazla2019reasoning, mossel2014asymptotic}.  Bayesian observational learning examples include \cite{banerjee1992simple, bikhchandani1992theory} and \cite{smith2000pathological}. These models are \textbf{sequential} in nature. They study herd behaviour.  As time passes, a player in turn observes the actions of previous agents and receives a private signal. Each agent has a one-off decision when she updates her posterior probability and takes an action. In some instances, the  $n$th agent may reach the truth as  $n \to \infty$. 

In Degroot learning, myopic updating occurs in each iteration. Agents in our setup have fixed weights but update their responses until consensus is reached. Recently there have been some experimental papers on the evidence of Degroot updating \cite{chandrasekhar2019testing, becker2017network}. Repeated averaging models are our base precisely because they capture the nature of interaction and learning in financial markets so compactly.  Players can observe previous choices but not the payoffs of their competitors. A more in depth discussion of learning in games would take us further away from our goal of studying the mathematical nature of interaction. The reader can consult \cite{fudenberg1998theory, kalai1994weak} for a game theoretic perspective. 

\textbf{Our contribution.} We introduce two different classes of learning models that converge to a consensus. Our interest is not in equilibrium but what process lead to it \cite{papadimitriou2018game,piliourasaamas, panageasopinion}. The first  introduces a feedback mechanism  (Section \ref{CDS1}, Theorem \ref{thm1}) where agents who are off the true ``hidden" volatility parameter feel a slight (even infinitesimally so) pull towards it along with the all the other ``random" chatter of the market. This model captures the setting where traders have access to an alternative trading venue or an information source provided by brokers and private message boards. The second model incorporates a market leader (e.g. Goldman Sachs) that is confident in its own internal metrics or is privy to client flow (private information) and does not give any weight to outside opinions (Section \ref{CDS2},Theorem \ref{thm2}).
Proving the convergence results (as well as establishing the exponentially fast convergence rates)  requires tools from discrete dynamical systems. We showcase as well as complement our theoretical results with experiments (e.g. Figures \ref{FigSim1}.a-\ref{FigSim1}.d), which for example show that if we move away from our models convergence is no longer guaranteed. 

We formalize the multi-dimensional analogues of our two models above  using Kronecker products (Section \ref{CDV}, Theorems \ref{thm3} and  \ref{thm4}). Thus our models show how a volatility curve could function as a global attractor given adaptive agents.  We conclude the paper by discussing future work on identifying necessary structural conditions on the shape of arbitrage free volatility curves.

\section{Model description}
In mathematical opinion dynamic models, agents take views of other agents into account before arriving at their own updated estimate. Agents can observe other agents' previous signals.

Degroot \cite{Degroot} was one of the early developers of such observational learning dynamics.  While simple, these models allow us to examine convergence to consensus.  In a sense, these type of models are called naive models, as agents can recall perfectly what the other players submitted in the previous round.  See the survey papers \cite{masuda2017random, acemoglu2011opinion, golublearning,noorazar2020recent}.

\subsection{Volatility Basics}
Investors have an initial opinion of the implied volatility, which subsequently gets updated after taking into account volatilities of other agents. A feedback mechanism aids the agents in arriving at the true volatility parameter.

At all times the focus is on a static picture of the volatility smile.  Within this static framework agents are updating their opinion of the true implied volatility. This updating occurs in a high-frequency sense. In an exchange setting, one can think of all bids and offers as visible to agents. The agents initially are unsure of the true value of the implied volatility, but by learning - and feedback - get to the true parameter. Our first attempt is a naive learning model common in social networks. Learning occurs between trading times. Thus our implicit assumption is that no transactions occur while traders are adjusting and learning each others quotes.

This rather peculiar feature is market practice. Trading happens at longer intervals than quote updating. This is as true for high frequency trading of stocks as it is for options markets. Quotes and prices - or rather vols - are changing more frequently than actual transactions.

Each dollar value of an option corresponds to an implied volatility parameter $\sigma(K,T) \in (0,1)$ that depends on strike and expiry. Implied volatility is quoted in percentage terms.

\begin{assumption}
	We have three types of players: agents/traders, brokers and leaders. Brokers give feedback to the traders.  The ability of agents to determine this feedback is their learning ability. Leaders are unknown and don't give feedback but their quotes are visible.
\end{assumption}
Each agent takes a weighted average of the all the agents' estimates of volatility at a particular strike and expiry.


\subsection{Naive Opinion Dynamics}

A first approach towards opinion dynamics is to assume each agent takes a weighted average of other agents' opinions and updates his own estimate of the volatility parameter for the next period, i.e., at time $t$, the opinion $x_t^i \in \mathbb{R}$ of the $i$-th agent is given by
	\begin{equation} \label{eqn:1p}
	x_t^i= \sum_{j=1}^n a_{ij} x_{t-1}^j, \hspace{1mm} t \in \mathbb{N},
	\end{equation}
where $x_{t-1}^j \in \mathbb{R}$ is the opinion of agent $j$ at time $(t-1)$ and $a_{ij} \geq 0$ denotes the opinion weights for the $n$ investors with $\sum_{j=1}^{n} a_{ij}=1$ and $a_{ii} > 0$ for all $1\leq i \leq n$. Define $X_t:=(x_t^1,\ldots,x_t^n)^{\top}$; then, the opinion dynamics of the $n$ agents can be written in matrix form as follows
	\begin{equation} \label{eqn:1}
	X_t=A X_{t-1},
	\end{equation}
where $A := a_{ij} \in \mathbb{R}^{n \times n}$ is a \emph{row-stochastic matrix}.

\begin{definition}[\textbf{consensus}]
The $n$ agents \eqref{eqn:1} are said to reach consensus if for any fixed initial condition $X_1 \in \mathbb{R}^n$,  $| x^i_t - x^j_t | \rightarrow 0$ as $t \rightarrow \infty$ for all $i,j \in \{1,\ldots n\}$.
\end{definition}

\begin{definition}[\textbf{consensus to a point}]
The $n$ agents \eqref{eqn:1} are said to reach consensus to a point if for any initial condition $X_1  \in \mathbb{R}^n$, $\lim_{t \to \infty} X_t =c \mathbf{1}_{n}$, where $\mathbf{1}_{n}$ denotes the $n \times 1$ vector composed of only ones and $c \in \mathbb{R}$. The constant $c$ is often referred to as the consensus value.
\end{definition}

For the opinion dynamics $\eqref{eqn:1}$, we introduce the following result by \cite{Degroot} (see also \cite{Salinelli2014} for definitions).

\begin{proposition}\label{Degroot}
Consider the opinion dynamics in equation \eqref{eqn:1}. If $A$ is aperiodic and irreducible, then for any initial condition $X_1  \in \mathbb{R}^n$ consensus to a point is reached. The consensus value $c$ depends on both the matrix $A$ and the initial condition $X_1$.
\end{proposition}

\begin{remark}
Proposition \ref{Degroot} implies that if the row stochastic opinion matrix $A$ is aperiodic and irreducible; then all the agents converge to some consensus value $c$. However, since $c$ depends on the unknown initial opinion $X_1$, the consensus value $c$ is unknown and, in general, different from the true volatility $\sigma(K,T)$. We wish to alleviate this and thus introduce two novel models.
\end{remark}

\section{Consensus (scalar agent dynamics)}\label{CDS}

In this section, we assume that the agents are able to learn how far off they are from the true volatility by informational channels in the marketplace. There are many avenues, platforms and private online chat rooms that provide quotes for option prices; some of these are stale and some are fresh. The agents' learning ability determines the quality of the feedback from all these sources. In reality, options are not traded on one exchange or platform. There are multiple venues and though there might be a dominant marketplace, the same instruments can be traded across different venues and locations. We aggregate all of this information in the form of feedback with learning ability. If agents are fast learners, they adjust their volatility  estimates quickly.

\subsection{Consensus with Feedback}
\label{CDS1}

We model this feedback by introducing an extra driving term into the opinion dynamics (\ref{eqn:1p}). An early model developed by Mizuno et al. \cite{mizuno2004traders} shares some similarities to ours. Traders use feedback from past behaviour. Our model is a discrete autoregressive process but the focus is on learning in high-frequency time \cite{mizuno2003analysis}. Furthermore, our model formalizes this in a more social and dynamical setup. In particular, we feedback the difference between the agents' opinion and the true volatility $\sigma(K,T)$ scaled by a \emph{learning coefficient} $\epsilon_i \in (0,1)$. We assume that $\sigma(K,T)$ is invariant, i.e., for some fixed $\bar{\sigma} \in (0,1)$, $\sigma(K,T) = \bar{\sigma}$ for some fixed strike $K$ and maturity $M$. Then, the new model is written as follows
\begin{equation} \label{eqn:2p}
x_t^i= \sum_{j=1}^n a_{ij} x_{t-1}^j + \epsilon_i(\bar{\sigma} - x_{t-1}^i),
\end{equation}
or in matrix form
\begin{align}\label{eqn:model1}
X_t=A X_{t-1} + \mathcal{E}(\bar{\sigma} \mathbf{1}_{n} - X_{t-1}),
\end{align}
where $\mathcal{E} := \text{diag}(\epsilon_1,\ldots,\epsilon_n)$. Then, we have the following result.
	
\begin{theorem}\label{thm1}
Consider the opinion dynamics \eqref{eqn:model1} and assume that $\epsilon_i \in (0,a_{ii})$, $i=\{1,\ldots,n\}$; then, consensus to $\bar{\sigma}$ is reached, i.e., $\lim_{t \to \infty} X_t = \bar{\sigma}\mathbf{1}_{n}$.
\end{theorem}
\begin{proof}
It is easy to verify that the solution $X_t$ of the difference equation (\ref{eqn:model1}) is given by
\begin{align}
X_{t+1}&=(A - \mathcal{E})^t X_1 + \text{$\sum$}_{j=0}^{t-1}(A - \mathcal{E})^j \mathcal{E}\bar{\sigma} \mathbf{1}_{n}, \hspace{1mm} t>1. \label{eqn:3}
\end{align}
By Gershgorin circle theorem, the spectral radius $\rho(A-\mathcal{E})<1$ for all $i$, $\epsilon_i < a_{ii}$. It follows that $\sum_{j=0}^{\infty}(A - \mathcal{E})^j \mathcal{E}\bar{\sigma} \mathbf{1}_{n} = (I_n - A + \mathcal{E})^{-1}\mathcal{E} \bar{\sigma} \mathbf{1}_{n}$, where $I_n$ denotes the identity matrix of dimension $n$, and $\lim_{t\to \infty} (A - \mathcal{E})^t=\mathbf{0}$, see \cite{Horn:2012:MA:2422911}. The matrix $A$ is row stochastic; then, $(I-A)\mathbf{1}_n=\mathbf{0}_n$, where $\mathbf{0}_n$ denotes the $n \times 1$ vector composed of only zeros. Hence, we can write $\mathcal{E} \mathbf{1}_n=(I_n-A)\mathbf{1}_n + \mathcal{E} \mathbf{1}_n$; and consequently $\mathbf{1}_n=(I_n-A + \mathcal{E})^{-1}\mathcal{E} \mathbf{1}_n$. It follows that
\begin{align*}
\lim_{t\to \infty} X_{t+1} &= \lim_{t\to \infty}(A - \mathcal{E})^t X_1+ \text{$\sum$}_{j=0}^{\infty}(A - \mathcal{E})^j \mathcal{E}\bar{\sigma} \mathbf{1}_{n}\\
&= \mathbf{0}_n + (I_n - A + \mathcal{E})^{-1}\mathcal{E}  \mathbf{1}_{n} \bar{\sigma} =  \mathbf{1}_{n}\bar{\sigma},
\end{align*}
and the assertion follows.
\end{proof}

\begin{corollary}\label{cor1}
Consensus to $\bar{\sigma}$ is reached exponentially with  convergence rate $\norm{A-\mathcal{E}}_\infty$, i.e., $\max_{i}\{|x_t^i - \bar{\sigma}|\} \leq \norm{A-\mathcal{E}}_\infty^{t-1} \norm{X_{1} - \bar{\sigma}\mathbf{1}_n}_\infty$, $i \in \{1,\ldots,n\}$, where $\norm{\cdot}_{\infty}$ denotes the matrix norm induced by the vector infinity norm. 
\end{corollary}
\begin{proof}
	Define the error sequence $E_{t-1}:= (X_{t-1} - \bar{\sigma}\mathbf{1}_n) \in \Real^n$. Then, from \eqref{eqn:model1}, the following is satisfied:
	\[
	\begin{array}{lll}
	E_{t} = X_{t} - \bar{\sigma}\mathbf{1}_n\\[1mm]
	\hspace{4.8mm} = AX_{t-1} + \mathcal{E}(\bar{\sigma}\mathbf{1}_n - X_{t-1}) - \bar{\sigma}\mathbf{1}_n\\[1mm]
	\hspace{4.8mm} = A(E_{t-1} + \bar{\sigma}\mathbf{1}_n) + \mathcal{E}(\bar{\sigma}\mathbf{1}_n - (E_{t-1} + \bar{\sigma}\mathbf{1}_n)) - \bar{\sigma}\mathbf{1}_n\\[1mm]
	\hspace{4.8mm} = (A-\mathcal{E})E_{t-1} + \bar{\sigma}(A-I_n)\mathbf{1}_n\\[1mm]
	\hspace{4.8mm} = (A-\mathcal{E})E_{t-1},
	\end{array}
	\]
	The last equality in the above expression follows from the fact that $(A-I_n)\mathbf{1}_n=0$, because $A$ is a stochastic matrix. The solution $E_t$ of the above difference equation is given by $E_t = (A-\mathcal{E})^{t-1}E_1$, where $E_1 = X_{1} - \bar{\sigma}\mathbf{1}_n$ denotes the initial error. Let $\norm{E_t}_{\infty} = \max_{i}(|e_t^i|)$, $i\in\{1,\ldots,n\}$, where $E_t = (e_t^1,\ldots,e_t^n)^T$. Note that exponential convergence of $\norm{E_t}_{\infty}$ implies exponential convergence of $E_t$ itself. Using the solution $E_t = (A-\mathcal{E})^{t-1}E_1$, the following can be written:
	\[
	\begin{array}{lll}
	\norm{E_t}_{\infty} =\norm{(A-\mathcal{E})^{t-1}E_1}_{\infty}\\[1mm]
	\hspace{11.5mm} \leq \norm{(A-\mathcal{E})}^{t-1}_{\infty} \norm{E_1}_{\infty},
	\end{array}
	\]
	where $\norm{(A-\mathcal{E})}_{\infty}$ denotes the matrix norm of $(A-\mathcal{E})$ induced by the vector infinity norm \cite{Horn:2012:MA:2422911}. The inequality $\norm{E_t}_{\infty} \leq \norm{(A-\mathcal{E})}^{t-1}_{\infty} \norm{E_1}_{\infty}$ implies exponential convergence if $\norm{(A-\mathcal{E})}_{\infty}<1$. Because $A=a_{ij}$ and $\mathcal{E}=\text{diag}(\epsilon_1,\ldots,\epsilon_n)$, we can compute $\norm{(A-\mathcal{E})}_{\infty}$ as $\norm{(A-\mathcal{E})}_{\infty}= \max_i\big( \sum_{j=1,j\neq i}^{n}|a_{ij}| + |a_i - \epsilon_i|   \big)$, $i \in \{1,\ldots,n  \}$. The matrix $A$ is stochastic, which implies $a_{ij}\geq 0$ and $\sum_{i=1}^{n} |a_{ij}| = 1$; therefore, under the conditions of Theorem \ref{thm1} (i.e., $\epsilon_i \in (0,a_{ii})$), $\norm{(A-\mathcal{E})}_{\infty}= \max_i\big( \sum_{j=1,j\neq i}^{n}|a_{ij}| + |a_i - \epsilon_i| \big) < 1$ and hence exponential convergence of the consensus error $E_t$ can be concluded with convergence rate given by $\norm{(A-\mathcal{E})}_{\infty}= \max_i\big( \sum_{j=1,j\neq i}^{n}|a_{ij}| + |a_i - \epsilon_i| \big)$.	
\end{proof}
\subsection{Random case}
Under suitable random conditions for the trust matrix $A$ and $\mathcal{E}$, we can still have consensus. In this case, the learning rates and and weights are independently and identically distributed from each iteration. However we need a condition to ensure convergence, namely that on average the learning rates are less than self-belief, condition. Since this is only in expectation, a probabilistic statement, there is some leeway on the learning rates being strictly less than self-belief $a_{ii}$ at time $t$.

\begin{theorem}\label{Randomthm1}
Consider the opinion dynamics 
\begin{align}\label{eqn:Randommodel1}
X_t=A_t X_{t-1} + \mathcal{E}_t(\bar{\sigma} \mathbf{1}_{n} - X_{t-1}),
\end{align}
where $A_t$ and $ \mathcal{E}_t$ are independent and identically distributed (iid). Furthermore suppose

\[
\mathbb{E}[\log \norm{A_t -\mathcal{E}_t}_\infty] < 0
\] 
then, consensus to $\bar{\sigma}$ is reached, i.e., $\lim_{t \to \infty} X_t = \bar{\sigma}\mathbf{1}_{n}$.
\end{theorem}

\begin{proof}
We rewrite the above iteration, by subtracting $\bar \sigma$, from both sides, dropping the one vector notation as the context is clear
\begin{align*}
X_t  -  \bar \sigma &=A_t X_{t-1} + \mathcal{E}_t(\bar{\sigma}  - X_{t-1}) - \bar \sigma\\
X_t  -  \bar \sigma &=A_t X_{t-1}   - A_t \bar \sigma+ \mathcal{E}_t\bar{\sigma}  - \mathcal{E}_t X_{t-1}\\
X_t  -  \bar \sigma &=(A_t -\mathcal{E}_t)  (X_{t-1}   - \bar \sigma)\\
Y_t  &=(A_t -\mathcal{E}_t)  Y_{t-1} \\
Y_t  &=B_t Y_{t-1},
\end{align*}
 where $Y_t=X_t  -  \bar \sigma$ and $B_t=A_t -\mathcal{E}_t$. We want to show $Y_t \to 0$. To this end, iterating the above recursion we arrive at
  \[
 Y_t =\underbrace{B_t B_{t-1}\cdots B_1}_{\mbox{iid matrices}} Y_0.
 \]
Taking norms on the above equation, gives us the following inequalities, understanding that we mean the $\norm{\cdot}_\infty$ norm:
\begin{align*}
 \norm{Y_t} &=\norm{B_t B_{t-1}\cdots B_1 Y_0}\\
  \norm{Y_t} &\leq\norm{B_t } \norm{B_{t-1}}\cdots \norm{B_1} \norm{Y_0}\\
 \log \norm{Y_t} &\leq \log \left(  \norm{B_t } \norm{B_{t-1}}\cdots \norm{B_1} \norm{Y_0}\right) \\
  \log \norm{Y_t} &\leq \log  \norm{B_t }  + \log \norm{B_{t-1}}+\cdots + \log \norm{B_1} + \log \norm{Y_0}\\
\norm{Y_t} &\leq \exp^{t \,\frac{\sum_{k=1}^{t}\log  \norm{B_k }}{t}}  \norm{Y_0}
\end{align*} 
The first inequality follows by sub-multiplicative property of matrix norms. Moreover, by the law of large numbers $\frac{1}{t}\sum_{k=1}^{t}\log  \norm{B_k }_\infty \longrightarrow \mathbb{E}[\log \norm{A_t -\mathcal{E}_t}_\infty]$, which is negative by assumption.  So the exponent ensures that, as the initial opinion $\norm{Y_0}_\infty < \infty$  is finite,

\[
\lim_{t \to \infty} \norm{Y_t}_\infty = 0.
\]
Consequently,  $Y_t \longrightarrow 0$ and every agent reaches consensus.
 
\end{proof}
Note we don't require the stronger condition that $\log \norm{A_t -\mathcal{E}_t}_\infty < 0,$ for all $t$. Unlike the deterministic case, the random case allows considerable flexibility. Neither self-belief $a_{ii} >0 $ nor positive learning $\epsilon_i$ is required for all times.  However, there must be some interaction and learning for beliefs to converge. As matrix products don't commute, if we were to follow the full recursion in any of our dynamics the result would be long matrix products. Random matrix products and dynamics are an active area of research not only in mathematics but also in physics \cite{diaconis1999iterated, crisanti2012products, bruneau2010infinite, garnerone2010typicality}. While the random case is certainly interesting, in this article our focus is on the first steps of modelling interaction and learning dynamics.

\subsection{Consensus with an unknown leader} \label{CDS2}
One criticism of model \eqref{eqn:model1} is that feedback, even if it is not perfect, has to be learned. In practice, there might not be a helpful mechanism that provides feedback. An alternative is to have an unknown leader embedded in the set of traders.  The agents are unsure who the leader is but by taking averages of other traders, they all arrive at the opinion of the leader. In markov chain theory, such behaviour is called an absorbing state. The leader guides the system to the true value. We assume that the \emph{identity} of the leader is unknown to all agents.\vspace{1mm}

Without loss of generality, we assume that the first agent (with corresponding opinion $x_t^1$) is the leader; it follows that $x_1^1=\bar{\sigma}$, $a_{1i}=0$, $i \in \{2,\cdots,n\}$, and $a_{11}=1$. Then, in this configuration, the opinion dynamics is given by
\arraycolsep=1.2pt\def\arraystretch{1.4}
\begin{equation} \label{leader_scalar}
X_t=A X_{t-1}, \hspace{1mm} A = \begin{pmatrix}
                                  1 & 0 & \ldots & 0 \\
                                  a_{21}& a_{22} & \ldots & a_{2n}\\
                                  \vdots & \vdots & \ldots & \vdots\\
                                  a_{n1}& a_{n2} & \ldots & a_{nn}
                                \end{pmatrix} =: \begin{pmatrix}
                                                   1 & \mathbf{0} \\
                                                   * & \tilde{A}
                                                 \end{pmatrix},
\end{equation}
with $a_{ij} \geq 0$, $\sum_{j=1}^{n} a_{ij}=1$, $a_{ii} > 0$ for all $1\leq i \leq n$, and for at least one $i$, $\sum_{j=2}^{n} a_{ij}<1$.
\begin{theorem}\label{thm2}
Consider the opinion dynamics \eqref{leader_scalar} and assume that the matrix $\tilde{A}$ is substochastic and irreducible. It holds that $\lim_{t \to \infty} X_t = \bar{\sigma}\mathbf{1}_{n}$, i.e., consensus to $\bar{\sigma}$ is reached.
\end{theorem}
\begin{proof}
Define the invertible matrix $M \in \mathbb{R}^{n \times n}$
\arraycolsep=1.2pt\def\arraystretch{1.2}
\[M :=\begin{pmatrix}
1&&\mathbf{0}\\
\mathbf{1}_{n-1}&&-I_{n-1}
\end{pmatrix}.\]
Introduce the set of coordinates $\tilde{X}_{t-1}:= MX_{t-1}$. Note that $\tilde{x}_{t-1}^1 = x_{t-1}^1$, $\tilde{x}_{t-1}^2 = x_{t-1}^1 - x_{t-1}^2, \ldots, \tilde{x}_{t-1}^{n} = x_{t-1}^1 - x_{t-1}^n$. Hence, if the error vector $e_{t-1} := (\tilde{x}_{t-1}^2,\ldots,\tilde{x}_{t-1}^n)^{\top} = \mathbf{0}_{n-1}$, then consensus to $x_{t}^1 = \bar{\sigma}$ is reached. Note that
\arraycolsep=.5pt\def\arraystretch{1.4}
\[
MAM^{-1} = \begin{pmatrix}
1&&*\\
\mathbf{0}&& \tilde{A}
\end{pmatrix} ,\]
where $\mathbf{0}$ denotes the zero vector of appropriate dimensions and $\tilde{A}$ as defined in (\ref{leader_scalar}). By construction, $\tilde{X}_{t-1}:= MX_{t-1} \rightarrow \tilde{X}_{t} = MX_{t} = MAX_{t-1} = MAM^{-1}\tilde{X}_{t-1}$; hence, the consensus error $e_{t}$ satisfies the following difference equation
\begin{align}\label{error_dyn}
\tilde{X}_{t} =  MAM^{-1}\tilde{X}_{t-1} = \begin{pmatrix} 1&&*\\ \mathbf{0}&& \tilde{A} \end{pmatrix} \tilde{X}_{t-1} \implies e_{t} = \tilde{A} e_{t-1},
\end{align}
and the solution of $e_{t}$ is then given by
$e_{t} = \tilde{A}^te_{1}$.

Because for at least one $i$, $\sum_{j=2}^{n} a_{ij}<1$ and $\tilde{A}$ is substochastic and irreducible, the spectral radius $\rho(\tilde{A})<1$, see Lemma 6.28 in \cite{Salinelli2014}; it follows that $\lim_{t \rightarrow \infty} \tilde{A}^t = \mathbf{0}$. Therefore, $\lim_{t \rightarrow \infty} e_{t} = \mathbf{0}$ and the assertion follows.
\end{proof}

\begin{corollary}\label{cor2}
Let $\norm{\cdot}_*$ denote some matrix norm such that $\norm{\tilde{A}}_*<1$ (such a norm always exists because $\rho(\tilde{A})<1$ under the conditions of Theorem \ref{thm2}). Then, consensus to $\bar{\sigma}$ is reached exponentially with the convergence rate given by $\norm{\tilde{A}}_*$, i.e. $\max_{i}\{\left| x_t^i - \bar{\sigma}\right|\} \leq C \norm{\tilde{A}}_*^{t-1} \norm{X_{1} - \bar{\sigma}\mathbf{1}_n}_\infty$, for $i \in \{1,\ldots,n\}$ and some positive constant $C \in \Real_{>0}$.
\end{corollary}

\begin{proof}
	See  Lemma 5.6.10 in \cite{Horn:2012:MA:2422911} on how to construct such a $\norm{\cdot}_*$. Now consider the consensus error $e_t$ defined in the proof of Theorem \ref{thm2}, which evolves according to the difference equation \eqref{error_dyn}. It follows that $e_t = \tilde{A}^{t-1}e_1$, where $e_1$ denotes the initial consensus error. Under the assumptions of Theorem \ref{thm2}, $\rho(\tilde{A})<1$. By Lemma 5.6.10 in \cite{Horn:2012:MA:2422911}, $\rho(\tilde{A})<1$ implies that there exists some matrix norm, say $\norm{\cdot}_*$, such that $\norm{\tilde{A}}_*<1$. We restate the error with norms and obtain $\norm{e_t}_\infty \leq \norm{\tilde{A}}^{t-1}_\infty \norm{e_1}_\infty$. Because all norms are equivalent in finite dimensional vector spaces (see Chapter 5 in \cite{Horn:2012:MA:2422911}), $\norm{e_t}_\infty \leq \norm{\tilde{A}}^{t-1}_\infty \norm{e_1}_\infty$ $\implies$ $\norm{e_t}_\infty \leq C \norm{\tilde{A}}^{t-1}_*  \norm{e_1}_\infty$ for some positive constant $C \in \Real_{>0}$. As $\norm{\tilde{A}}_*<1$, the norm of the consensus error $\norm{e_t}_\infty$ converges to zero exponentially with rate $\norm{\tilde{A}}_*$.
\end{proof}

\section{Consensus (vectored agent dynamics)}\label{CDV}
In this section, we suppose that agents have beliefs over a range of strikes.  Thus, each agent's opinion of the volatility curve is a vector with each entry corresponding to a particular strike. Typically, in markets, options are quoted for at-the-money (atm) $K=S_0$ and for two further strikes left of and right of the atm level. Here, we examine the case of $k$ strikes and $n$ agents, i.e., each agent $i$ now has $k$ quotes for $k$ different moneyness levels. In this configuration, the true volatility is $\bar{\sigma} :=[\sigma_1,\ldots,\sigma_k]^{\top}\in\mathbb{R}^k$. See figure \ref{fig:smile} (b).

\subsection{Consensus with Feedback}

Again, we assume that each agent takes a weighted average of other agents' opinions and updates its volatility estimate \emph{vector} for the next period, i.e., at time $t$, the opinion $x_t^i \in \mathbb{R}^k$ of the $i$-th agent is given by
	\begin{equation} \label{eqn:1pm}
	x_t^i= \sum_{j=1}^n a_{ij} x_{t-1}^j + \epsilon_i(\bar{\sigma} - x_{t-1}^i), \hspace{1mm} t \in \mathbb{N},
	\end{equation}
where $\epsilon_i \in (0,1)$ denotes the \emph{learning coefficient} of agent $i$, $x_{t-1}^j \in \mathbb{R}^k$ is the opinion of agent $j$ at time $(t-1)$, and $a_{ij} \geq 0$ denotes the opinion weights for the $n$ investors with $\sum_{j=1}^{n} a_{ij}=1$ and $a_{ii} > 0$ for all $1\leq i \leq n$. In this case, the stacked vector of opinions is $X_t:=(x_t^1,\ldots,x_t^n)^{\top}$, $X_t \in \mathbb{R}^{kn}$. The opinion dynamics of the $n$ agents can then be written in matrix form as follows
	\begin{equation} \label{eqn:1m}
	X_t=(A \otimes I_k) X_{t-1} + (\mathcal{E} \otimes I_k)(\mathbf{1}_{n} \otimes \bar{\sigma}  - X_{t-1}),
	\end{equation}
where $A = a_{ij} \in \mathbb{R}^{n \times n}$ is a \emph{row-stochastic matrix}, $\mathcal{E} = \text{diag}(\epsilon_1,\ldots,\epsilon_n)$, and $\otimes$ denotes Kronecker product. We have the following result.

\begin{theorem}\label{thm3}
Consider the opinion dynamics in \eqref{eqn:1m} and assume that $\epsilon_i \in (0,a_{ii})$, $i=\{1,\ldots,n\}$; then, consensus to $\mathbf{1}_{n} \otimes \bar{\sigma}$ (with $\bar{\sigma} = [\sigma_1,\ldots,\sigma_k]^{\top} \in \mathbb{R}^k$) is reached, i.e., $\lim_{t \to \infty} X_t = \mathbf{1}_{n} \otimes \bar{\sigma}$.
\end{theorem}
\begin{proof} Define the error sequence $e_{t-1}:= X_{t-1} - (\mathbf{1}_{n} \otimes \bar{\sigma})$. Note that $e_{t-1}= \mathbf{0}$ implies that consensus to $(\mathbf{1}_{n} \otimes \bar{\sigma})$ is reached. Given the opinion dynamics \eqref{eqn:1m}, the evolution of the error $e_{t-1}$ satisfies the following difference equation
\begin{align*}
e_{t}&= ((A - \mathcal{E}) \otimes I_k) X_{t-1} + ((\mathcal{E} \otimes I_k)-I_{kn})(\mathbf{1}_{n} \otimes \bar{\sigma})\\
     &= ((A - \mathcal{E}) \otimes I_k) e_{t-1} - (\mathbf{1}_{n} \otimes \bar{\sigma}) + (A \otimes I_k) (\mathbf{1}_{n} \otimes \bar{\sigma})\\
     &= ((A - \mathcal{E}) \otimes I_k) e_{t-1} + ((A-I_n)\mathbf{1}_{n} \otimes \bar{\sigma}).
\end{align*}
It is easy to verify that, because $A$ is stochastic, $(A-I_n)\mathbf{1}_{n} = \mathbf{0}_n$. Then, the error dynamics simplifies to
\begin{align}\label{error_dyn_Feedback}
e_{t} = ((A - \mathcal{E}) \otimes I_k) e_{t-1},
\end{align}
and consequently, the solution $e_{t}$ of \eqref{error_dyn_Feedback} is given by  $e_{t} = ((A - \mathcal{E}) \otimes I_k)^t e_1$. By properties of the Kronecker product and Gershgorin's circle theorem, the spectral radius $\rho(A - \mathcal{E})<1$ for $\epsilon_i \in (0,a_{ii})$. It follows that $\lim_{t \rightarrow \infty} ((A - \mathcal{E}) \otimes I_k)^t = \mathbf{0}$, see \cite{Horn:2012:MA:2422911}. Therefore, $\lim_{t \rightarrow \infty} e_{t} = \mathbf{0}_{kn}$ and the assertion follows.
\end{proof}

\begin{corollary}\label{cor3}
	Consensus to $\bar{\sigma}$ is reached exponentially with the convergence rate given by $\norm{(A - \mathcal{E}) \otimes I_k)}_\infty$, i.e., $\norm{X_t - (\mathbf{1}_n\otimes \bar{\sigma})}_\infty \leq \norm{(A - \mathcal{E}) \otimes I_k)}_\infty^{t-1} \norm{X_{1} - (\mathbf{1}_n\otimes\bar{\sigma})}_\infty$.
\end{corollary}

The proof of the above result is very similar to previous corollaries and is omitted.

\subsection{Consensus with an unknown leader}
Similarly to the scalar case; here, we assume that there is a leader driving all the other agents through the opinion matrix $A$. Again, without loss of generality, we assume that the first agent (with corresponding opinion $x_t^1 \in \mathbb{R}^k$) is the leader, $x_1^1=\bar{\sigma} =[\sigma_1,\ldots,\sigma_k]^{\top} \in \mathbb{R}^k$, $a_{1i}=0$, $i \in \{2,\cdots,n\}$, and $a_{11}=1$. Then, in this configuration, the opinion dynamics is given by
\arraycolsep=1.2pt\def\arraystretch{1.2}
\begin{equation} \label{leader_vectored}
X_t=(A \otimes I_k) X_{t-1}, \hspace{.5mm} A = \begin{pmatrix}
                                  1 & 0 & \ldots & 0 \\
                                  a_{21}& a_{22} & \ldots & a_{2n}\\
                                  \vdots & \vdots & \ldots & \vdots\\
                                  a_{n1}& a_{n2} & \ldots & a_{nn}
                                \end{pmatrix} =: \begin{pmatrix}
                                                   1 & \mathbf{0} \\
                                                   * & \tilde{A}
                                                 \end{pmatrix},
\end{equation}
with $a_{ij} \geq 0$, $\sum_{j=1}^{n} a_{ij}=1$, $a_{ii} > 0$ for all $1\leq i \leq n$, and for at least one $i$, $\sum_{j=2}^{n} a_{ij}<1$.

\begin{theorem}\label{thm4}
Consider the opinion dynamics \eqref{leader_vectored} and assume that the matrix $\tilde{A}$ is substochastic and irreducible; then, consensus to $\mathbf{1}_{n} \otimes \bar{\sigma}$ is reached, i.e., $\lim_{t \to \infty} X_t = \mathbf{1}_{n} \otimes \bar{\sigma}$.
\end{theorem}
The proof of Theorem \ref{thm4} follows the same line as the proof of Theorem \ref{thm2} and it is omitted here.

\begin{corollary}\label{cor4}
Let $ \norm{\cdot}_*$ denote some matrix norm such that $\norm{\tilde{A}}_*<1$, then consensus to $\bar{\sigma}$ is reached exponentially with convergence rate $\norm{\tilde{A}\otimes I_k}_*$, i.e. $\norm{X_{t} - (\mathbf{1}_n \otimes \bar{\sigma})}_\infty \leq C \norm{\tilde{A}\otimes I_k}_*^{t-1} \norm{X_{1} - (\mathbf{1}_n \otimes \bar{\sigma})}_\infty$, for some positive constant $C \in \Real_{>0}$.
\end{corollary}

\section{Numerical Simulations}

\begin{figure}[t]
	\centering
	\includegraphics[scale=0.6]{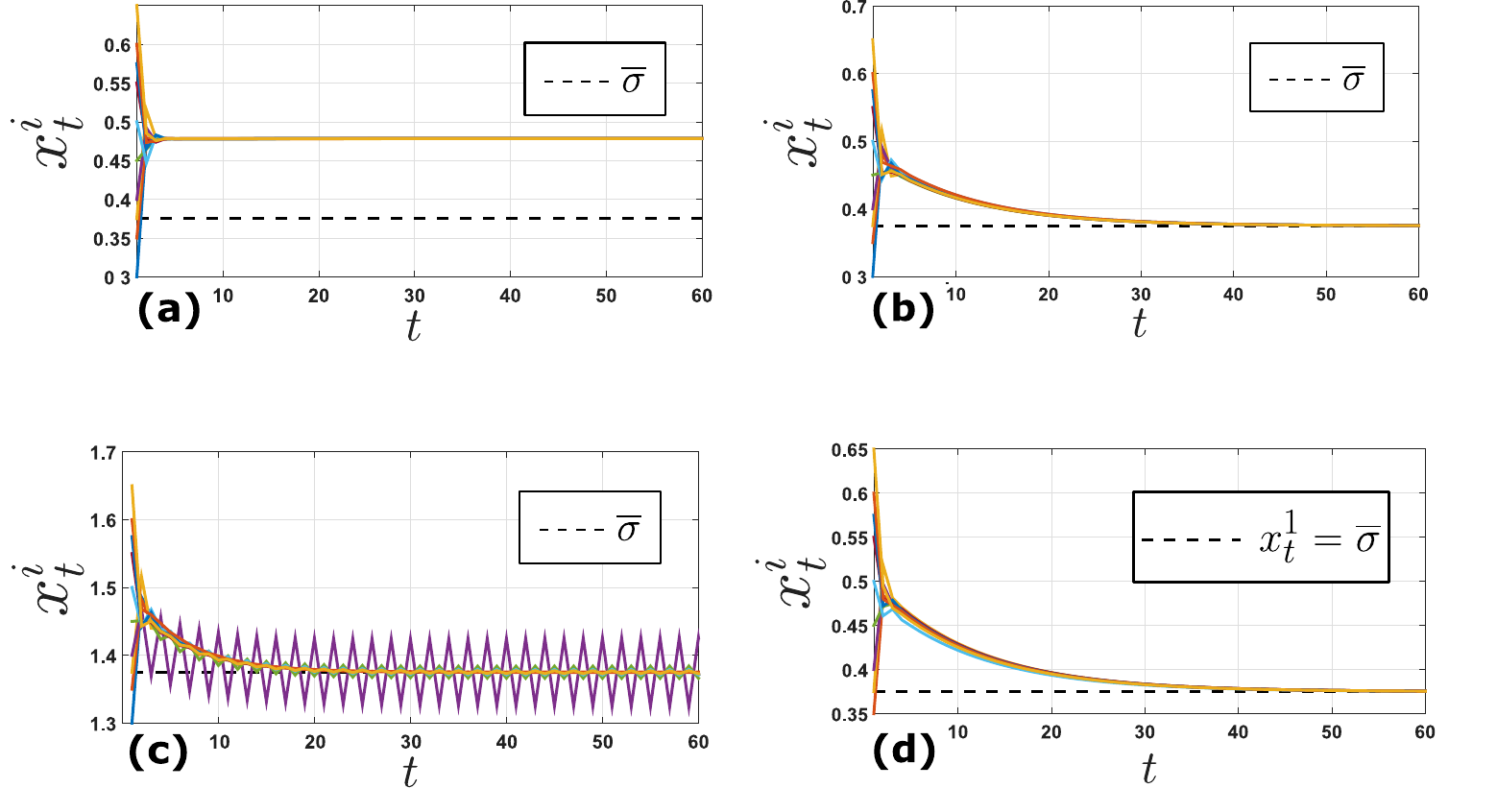}
	\caption{Evolution of the agents' dynamics \eqref{eqn:model1}: \textbf{(a)} without learning, \textbf{(b)} with learning and $\epsilon_i$ satisfying the conditions of Theorem \ref{thm1}, \textbf{(c)} with learning and $\epsilon_i$ \emph{not} satisfying the conditions of Theorem \ref{thm1}, and \textbf{(d)} Evolution of the agents' dynamics with leader \eqref{leader_scalar}.}\label{FigSim1}
\end{figure}
Consider the opinion dynamics with feedback \eqref{eqn:model1} with ten agents (i.e., $n = 10$) $\bar{\sigma} = 0.375$, and initial condition 
\[
X_1 = (0.3, 0.35, 0.37, 0.4, 0.45, 0.5, 0.55, 0.57, 0.6, 0.65)^{\top}.
\]
In both exchange-based and OTC markets it is easy to ascertain who the main market-makers are for options on single stock or commodity \cite{gueant2016financial, bouchaud2018trades}. Option market-makers are usually investment banks and big trading houses. In this sense, the number of players is not large and thus the models developed always have a finite number of agents, $N=10$.

Figure \ref{FigSim1} depicts the obtained simulation results for different values of the learning parameters $\epsilon_i$, $i=1,\ldots,10$. Specifically, Figure \ref{FigSim1} (a) shows results without learning, i.e, $\epsilon_i = 0$ (here there is no consensus to $\bar{\sigma}$), Figure \ref{FigSim1} (b) depicts the results for $\epsilon_i = 0.9a_{ii}$. As stated in Theorem \ref{thm1}, consensus to $\bar{\sigma}$ is reached. Figure \ref{FigSim1}(c) shows results for $\epsilon_i = 0.9a_{ii} + 0.94 b_i$ with $b_4 = 1$ and $b_i=0$ otherwise, $i=1,\ldots,10$. Note that, in this case, the value of $\epsilon_4$ violates the condition of Theorem \ref{thm1} (i.e., $\epsilon_4 \notin (0,a_{44})$) and, as expected, consensus is not reached. Next, consider the opinion dynamics with leader \eqref{leader_scalar} with $n = 10$ and initial condition 

\[
X_1 = (\bar{\sigma}, 0.35, 0.37, 0.4, 0.45, 0.5, 0.55, 0.57, 0.6, 0.65)^{\top}. 
\]

For the leader case, the opinion weights matrix is constructed by replacing the first row of $A$ by $(1,0,\ldots,0)$. The corresponding matrix $\tilde{A}$ (defined in \ref{leader_scalar}) is substochastic and irreducible, and $\sum_{i=2}^{i=10} a_{ij}<1$, $j=1,\ldots,10$. Hence, all the conditions of Theorem \ref{thm2} are satisfied and consensus to $\bar{\sigma} = 0.375$ is expected. Figure \ref{FigSim1}(d) shows the corresponding simulation results. Finally, Figure \ref{FigSim3} shows the evolution of the vectored opinion dynamics (\ref{eqn:1m}) with $n = 10$ and $k=3$ (i.e., ten three dimensional agents), matrix $A$ as in the case with feedback, (vectored) volatility $\bar{\sigma} = (0.67, 0.22, 0.88)^{\top}$, learning parameters $\epsilon_i = 0.9a_{ii}$ for $a_{ii}$ as in $A$, and initial condition $\mathbf{1}_k \otimes X_1$ with $X_1$ as in the first experiment above.


\begin{figure}[t]
	\centering
	\includegraphics[scale=0.35]{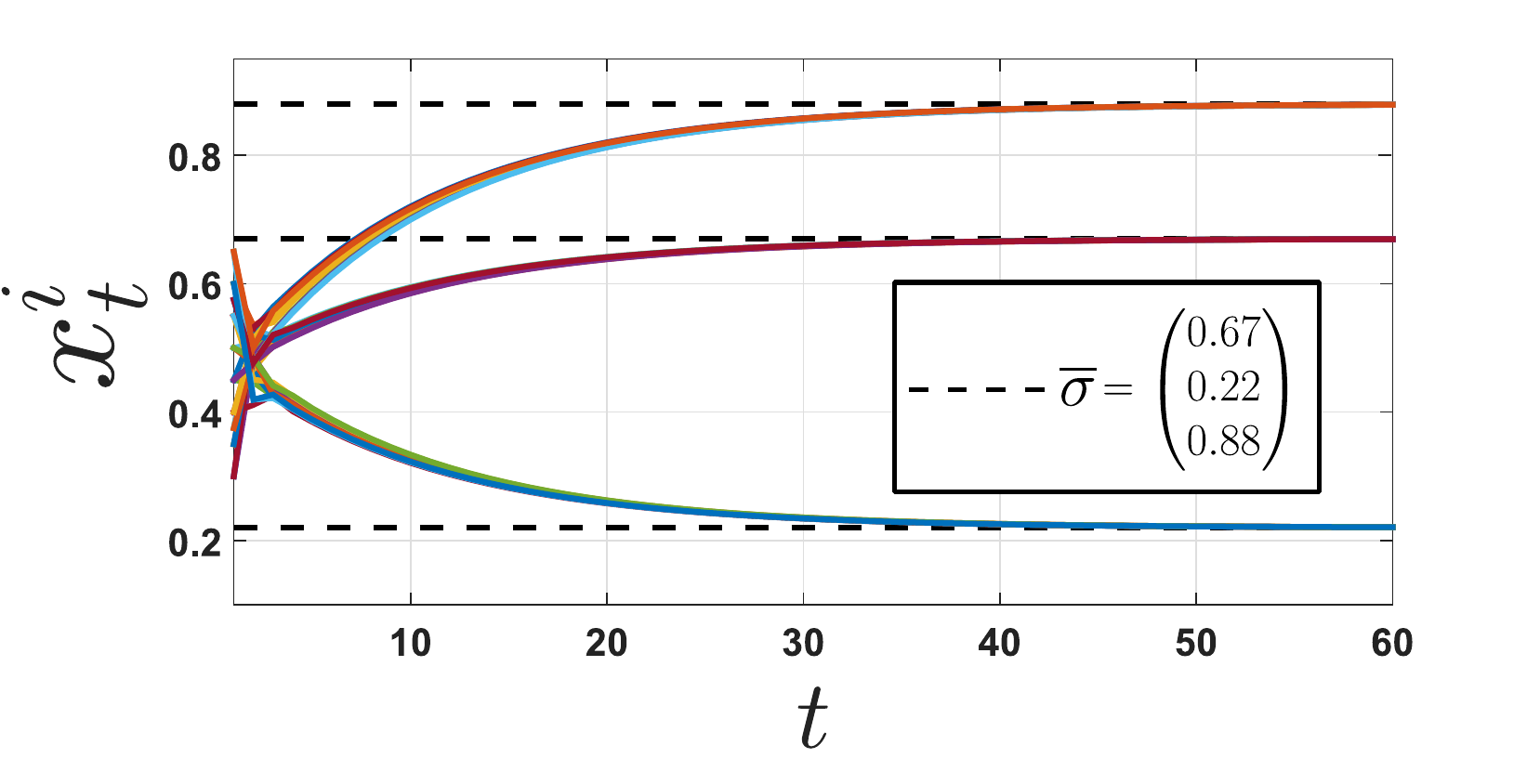}
    \caption{Evolution of the multidimensional agents' dyna\-mics with learning \eqref{eqn:1m}.}\label{FigSim3}
\end{figure}

\section{ Arbitrage Bounds}
We have taken the true volatility parameter as exogenous to our models. Our only requirement is that there is no static arbitrage, by which we mean that all the quotes in volatility which translate to option prices are such that one cannot trade in the different strikes to create a profit. Checking whether a volatility surface is indeed arbitrage free is non-trivial, nevertheless some sufficient conditions are well known~\cite{carr2005note, gatheral2014arbitrage, tehranchi2016uniform}. As long as the volatility surface satisfies them our analysis implies global stability towards an arbitrage free smile.

We parameterize the volatility function (assuming expiry $T \mbox{ and }S_0$ are fixed) and denote the option price as
\[\overline{BS}(K,\sigma(K)) \triangleq BS(S_0,K,T,\sigma(K)).\]  Our attention is on varying $K$, to ensure no static arbitrage. We assume that the $\sigma(K)$ translates into unique call option dollar prices, which follows from the strictly positive first derivative of the option price with respect to $\sigma$.

\begin{itemize}
\item {\bf Condition 1: (Call Spread)}
For $0 < K_1 \leq K_2$, we have 
$
\overline{BS}(K_1,\sigma(K_1))\geq \overline{BS}(K_2, \sigma(K_2)).
$

\item {\bf Condition 2: (Butterfly Spread)}  For $0 < K_1 < K_2 < K_3$, \\
$
\overline{BS}(K_1, \sigma(K_1))  + \frac{K_2-K1}{K_3-K_2}\times \overline{BS}(K_3, \sigma(K_3)) \geq \dfrac{K_3-K1}{K_3-K_2}\times \overline{BS}(K_2, \sigma(K_2)).
$

\end{itemize}

How these arbitrage-free curve volatility conditions are developed is not an easy task: see an account by \cite{roper2010arbitrage, lee2005implied}.  Delving into this topic would take us further into stochastic analysis and away from the focus of this paper.

\section{Connections and Conclusion}

Recently, there has been some rather interesting work on the intersection of computer science and option pricing. Demarzo et al. \cite{demarzo2006online} showed how to use efficient online trading algorithms to price the current value of financial instruments, deriving both upper and lower bounds using online trading algorithms.
Moreover, Abernethy et al. \cite{abernethy2012minimax,Frongillo2013hedge} developed Black-Scholes price as sequential two-player zero-sum game.  Whilst these papers made an excellent start to bridge the gap between two different academic communities - mainly mathematical finance and theoretical computer science - they do not address the reality of volatility smiles and trading. 
Our contribution can be viewed as making these connections more concrete. The smile itself is a conundrum and there have even been articles questioning whether it can be solved \cite{ayache2004can}.  The traditional way from the ground up is to develop a stochastic process for the volatility and asset price, possibly introducing jumps or more diffusions through uncertainty \cite{kamal2010implied, kyprianou2006exotic}.  Such models have been successfully developed, but the time is ripe to incorporate multi-agent models with arbitrage free curves. 

Combining learning agents in stochastic differential equation models \cite{schweizer2008arbitrage}, such as the Black-Scholes model, is an exciting proposition. Moreover, opinion dynamics as a subject on its own has been studied quite extensively. Recent references that present an expansive discussion in computer science are \cite{mosselopinion, panageasopinion}. Econophysics is the right community to develop new models. After all, there is no attachment to utilities of players or stochastic volatility models so beloved in the mathematical finance community. Free from these shackles, researchers can use a range of tools and techniques to build  more sophisticated models. Moreover, there is no restriction or debate on continuous or discrete time. While our framework is discrete, a continuous time could perhaps show a way forward to incorporate models from mathematical finance and financial economics \cite{nadtochiy2017robust, davis2007range, shafer2019game}.  The technical issues in random matrix products, briefly discussed in this paper, assure us that much more work needs to be done on the modelling and mathematical front. For example, the matrices $A$ and $\mathcal{E}$ can be dependent with correlation decreasing in time. The random case contraction would still hold. 

In this paper, we  introduce models of learning agents in the context of option trading. A key open question in this setting is how the market comes to a consensus about market volatility, which is reflected in derivative pricing through the Black-Scholes formula. The framework we have established allows us to explore other areas.
Thus far, we took the smile as an exogenous object, proving convergence to equilibrium beliefs. A natural step forward would be to look at the beliefs as probability measures, where each measure corresponds to a different option pricing model. Our learning models focus on interaction between agents. Actually, agents can be interpreted as algorithms. Each algorithm corresponding to a particular belief of a pricing model.

\section*{Acknowledgements}
The authors would like to thank Elchanan Mossel, Ioannis Panageas, Ionel Popescu and JM Schumacher for fruitful discussions. Tushar Vaidya would like to acknowledge a SUTD Presidential fellowship.
Carlos Murguia would like to acknowledge the National Research Foundation (NRF), Prime Minister's Office, Singapore, under its National Cybersecurity R\&D Programme (Award No. NRF2014NCR-NCR001-40) and administered by the National Cybersecurity R\&D Directorate.
Georgios Piliouras would like to acknowledge SUTD grant SRG ESD 2015 097 and MOE AcRF Tier 2 Grant  2016-T2-1-170.
\bibliographystyle{siam}
\bibliography{arxivV4}

\end{document}